\documentclass[reprint,prl,amsmath,amssymb,aps,superscriptaddress,nofootinbib]{revtex4-2}
\usepackage[colorlinks=true,allcolors=blue]{hyperref}
\usepackage{graphicx}
\usepackage{dcolumn}
\usepackage{mathrsfs}
\usepackage[capitalize]{cleveref}

\usepackage{booktabs}
\usepackage{amsthm}
\usepackage{xcolor}
\definecolor{turquoise}{HTML}{39A78E}

\newcommand{\F}{\mathcal{F}}
\usepackage{bm}
\usepackage{bbm}
\usepackage{algorithm}
\usepackage{algpseudocode}
\usepackage{tcolorbox}
\usepackage{verbatim}
\usepackage[normalem]{ulem}

\newcommand{\Ze}{\mathcal{Z}}
\newcommand{\X}{\mathcal{X}}
\newcommand{\Ce}{\mathcal{C}}
\newcommand{\C}{\mathcal{C}}
\newcommand{\T}{\mathcal{T}}
\newcommand{\eps}{\varepsilon}
\def\q{\bar q}
\def\MC{\mathcal{M}}
\def\argmin{\operatorname*{arg\,min}}
\def\Cres{\mathcal{C}_{\mathrm{res}}}

\newtheorem{thm}{Theorem}[section]

\theoremstyle{definition}

\newcommand{\bs}[1]{\boldsymbol{#1}}

\begin{document}
\title{Strengthened second law for periodic processes}
\author{Jake Schaefer$^{*}$}
\affiliation{Program in Applied and Computational Mathematics, Princeton University, Princeton, NJ 08544, USA}
\author{Ben Ansbacher$^{*,\dagger}$}
\affiliation{Santa Fe Institute, Santa Fe, NM 87501, USA}
\author{Jan Korbel}
\affiliation{Complexity Science Hub, Metternichgasse 8, 1030, Vienna, Austria}
\author{David Wolpert}
\affiliation{Santa Fe Institute, Santa Fe, NM 87501, USA}
\affiliation{Complexity Science Hub, Metternichgasse 8, 1030, Vienna, Austria}
\date{\today}

\begin{abstract}
We derive a lower bound on the entropy production of any physical system that evolves as a periodic process.  
Crucially, this lower bound is completely independent of the physical mechanism implementing the periodic dynamics; it follows purely from the definition of entropy production and the fact that whatever the physical mechanism is, it repeats in each period.
In particular, the lower bound holds in both Markovian and non-Markovian systems, in systems that obey or violate (local) detailed balance, and even if the underlying physical mechanism uses currently unknown laws of physics. 
Moreover, the bound is strictly positive whenever the dynamics is not logically invertible.
In this breadth of applicability, our bound is a strict strengthening of the second law. We illustrate our bound for spins coupled to two heat reservoirs and for finite automata.
\end{abstract}

\maketitle
\begingroup
\renewcommand{\thefootnote}{\fnsymbol{footnote}}
\footnotetext[1]{These authors contributed equally to this work.}
\footnotetext[2]{Contact author: \href{mailto:bdansbacher@gmail.com}{bdansbacher@gmail.com}}
\endgroup

\emph{Introduction---}
Many physical systems are periodic in the sense that whatever physical mechanism is implementing the dynamics, it repeats in each period, even though the state of the system need not return to where it started after each period. For example, synchronous digital computers advance their internal state once per clock cycle \cite{wolpert_stochastic_2019}; chemical reaction networks are driven by periodically modulated reservoirs \cite{lawson2020chemicalresonance}; molecular motors are periodically rotated by other motors to synthesize ATP \cite{mishima2025efficiently}; and synthetic genetic oscillators \cite{elowitz2000repressilator} and the circadian \emph{KaiC} phosphorylation cycle \cite{dong2008cyanobacterium} display the same structure. Whenever such a process is implemented physically, it carries a thermodynamic cost: work is supplied, heat is dissipated, particles are exchanged. Additionally, these processes are almost by definition, out of equilibrium. Such costs can be quantified by stochastic thermodynamics at the mesoscopic scale and arbitrarily far from equilibrium \cite{vandenbroeck2015ensemble,seifert2008principles,Seifert_2012}, assigning work, heat, and entropy production (EP) to individual fluctuating trajectories.

Most of the previous analyses of EP in periodically driven systems---small heat engines, molecular pumps, quantum dots, Brownian clocks \cite{barato2017periodic,Proesmans_2016,barato2016brownianclock}---proceed by deriving fluctuation theorems for the relevant currents, or use inference methods to estimate EP \cite{Barato2024}. The resulting bounds depend on the specific physical setting, since they rely on the heat flows realizing the protocol and typically assume the dynamics within each period is a continuous-time Markov chain (CTMC). A separate line of work considers oscillators whose cycle time itself fluctuates \cite{gopal2024electronicclock, horowtiz_oscillations}.

In this letter, we consider a complementary, alternative setting: a general process process represented by a map $G$ that describes how the initial state is transformed to the final state over a period of time $t$. This map $G$ is then repeated, modeling a periodic process. We emphasize that $G$ need not be Markovian within the single period. Finally, because it is reapplied to any any distribution that is the final distribution of the previous cycle, the state distribution $p_t$ need not be periodic even though the dynamics generating it is.

To this end we employ the framework of the \emph{mismatch cost} (MMC) decomposition \cite{mmc_2017,premmc2021}, in which a broad family of thermodynamic costs, EP among them, splits into a residual term---corresponding to the minimum, inevitable cost---plus the excess cost incurred when the process is initialized away from the cost-minimizing distribution, which we call the \emph{prior}. The prior encodes everything about the implementation that is thermodynamically relevant, so in general it depends on the physical details, e.g., how heat flows during the process, and is typically unknown; recent applications \cite{message_passing_2026, yadav2024mismatch, yadav2026entropy, kolchinsky_circuits, Ouldridge_2023, tasnim2023stochastic} have had to choose illustrative forms for it. Ref. \cite{yousef2026minimax} considered the `minimax' problem of choosing the prior that minimizes the worst-case cost, providing a prior-independent lower bound on worst-case EP. 

The key observation for periodic processes is that, because the \emph{same} physical process is carried out in every period, the prior is the same at every application of $G$, even though the distribution it acts on changes from cycle to cycle. The MMC accumulated over $N$ periods therefore depends on a single unknown prior. We thus minimize the accumulated MMC over all priors, and show the minimizer---the \emph{periodic-optimal prior}---is simply the time average of the trajectory, and that the resulting bound on thermodynamic cost depends only on $p_0$ and $G$; it makes no reference to the prior for the single period and thus the physical implementation of the process. It only relies on the definition of EP, and on the fact that the process is periodic. So it holds for both Markovian and non-Markovian dynamics, independent of whether local detailed balance holds, and even if the underlying mechanism exploits currently unknown laws of physics. This contrasts starkly with the myriad recent bounds on EP which have been derived in the stochastic thermodynamics literature, which almost universally assume Markovian dynamics (or in the case of finite reservoirs, Hamiltonian dynamics~\cite{Ptaszynski&esposito,jarzynski_infoprocess}), 
and typically require conservative forces and local detailed balance~\cite{speedlimits_fisher,speedlimits_ito, speedlimits_saito, horowtiz_oscillations,seifertTURs,saitoKURs,horowitzTURs, yoshihikoTURs}. 

In addition to these features, the bound, which we call the \emph{universal dissipation} is strictly positive whenever $G$ is not logically invertible. In light of the foregoing, this bound is a \emph{strict strengthening of the second law for any periodic process}. 

We illustrate the application of the bound for spins coupled to two baths, where we investigate what fraction of EP is due to the universal dissipation, and for deterministic finite automata (DFAs), whose operation is periodic by construction.
\newline

\emph{Mismatch cost---}
Consider a finite state space $\X$ of size $|\X|$ with distributions $p\in\Delta_\X$, where $\Delta_\X$ is the set of distributions over $|\X|$ outcomes. We define the one-cycle dynamics over an period $[0,\tau]$ by considering a stochastic matrix $G$ such that $p_{\tau}=Gp_0$. We assume the same $G$ is applied once every cycle. $G$ specifies only the conditional final distribution over a cycle; it does not have to be a time-ordered exponential of a Markovian dynamics governed by a (possibly time-dependent) rate matrix. 

The total EP over one cycle can be written as \cite{seifert2008principles}
\begin{equation}\label{eq:ep}
    \Sigma(p_0) = S(Gp_0)-S(p_0)+\Phi(p_0),
\end{equation}
where $S(p)=-\sum_x p(x)\log p(x)$ and $\Phi$ is the expected entropy flow out of the system (proportional to the heat flow under local detailed balance). The modern understanding of the second law states $\Sigma\geq 0$. The entropy flow  $\Phi$ depends on intermediate states in the transition $p_0\mapsto Gp_0$, and thus on the physical mechanism, whereas the entropy drop $S(Gp_0)-S(p_0)$ does not. Therefore the thermodynamic cost of a process generally depends on both the initial distribution $p_0$ and the details of its implementation. This distinction is the starting point for the mismatch-cost framework described below.

More generally, consider any cost function written in a similar form as above
\begin{equation}\label{eq:cost_form1}
    \C(p_0) = S(Gp_0)-S(p_0)+F(p_0),
\end{equation}
with $F$ a linear functional (possibly accumulated along the trajectory). Total EP, the change in nonequilibrium free energy, and nonadiabatic EP are all of this form (the associated linear functionals $F$ for these costs can be found in Section VII of Ref \cite{premmc2021}). Let $q\in\argmin_{p_0}\C(p_0)$ be the \emph{prior}. Then \cite{mmc_2017,premmc2021}
\begin{equation}\label{eq:cost}
    \C(p_0) = \underbrace{D(p_0\|q)-D(Gp_0\|Gq)}_{\MC_G(p_0\|q)} + \Cres,
\end{equation}
where $D$ is the Kullback--Leibler (KL) divergence. The first two terms are the MMC, which satisfies $\MC_G(p_0\|q)\geq0$ by the data-processing inequality for KL divergence, with equality at $p_0=q$; $\Cres:=\C(q)$ is the \emph{residual cost}. Whenever $\Cres\geq0$, as for EP, the MMC lower bounds the cost, and it can account for a substantial fraction of EP \cite{yadav2024mismatch, yadav_comm_2025, yadav2026entropy}. The residual cost can also be estimated or bounded using tools from stochastic thermodynamics like speed-limit theorems and thermodynamic uncertainty relations. The analysis becomes more complicated if there are multiple states of $\X$ that all evolve to the same state under $G$ (so-called ``islands"). The bounds derived below still hold even if that's the case though \cite{kolchinsky_circuits}(SM~\ref{ap:islands}). 

Since the physical process is identical every cycle, so is its prior \cite{Ouldridge_2023}. Starting from $p_0$, the \emph{periodic} MMC accumulated over $N$ cycles is
\begin{equation}\label{eq:periodic_mmc}
    \MC_N(q) = \sum_{t=0}^{N-1}\bigl[D(p_t\|q)-D(p_{t+1}\|Gq)\bigr],\quad p_t=G^tp_0,
\end{equation}
where we now make the dependence on $q$ explicit. The prior is fixed by the implementation, so changing the implementation changes $\MC_N$. This raises a natural question: \emph{which prior minimizes the accumulated MMC, and what is the resulting minimal EP of a periodic process?}
\newline

\emph{Main result---}
For any $G$ and $p_0$, the accumulated MMC is minimized by the \emph{universal prior}
\begin{equation}\label{eq:qstar}
    \q := \frac1N\sum_{t=0}^{N-1}p_t \;=\; \argmin_{q\in\Delta_\X}\MC_N(q),
\end{equation}
the temporal average of the marginal distributions of the system's state \footnote{Note that the support of $\q$ is the union of the supports of $p_0,\ldots,p_{N-1}$, whereas any prior arising from minimizing \eqref{eq:cost_form1} has full support unless $G$ is deterministic \cite{mmc_2017}; the bound may therefore not be saturable by any physical implementation. This bound is a purely mathematical statement about the minimum.}. The proof and conditions for uniqueness are given in SM~\ref{ap:weighted_pop}, where the result is also related to earlier work. If $G$ is irreducible with a single island and $\q$ has full support, $\q$ is the \emph{unique} minimizer; with several islands the minimizers are exactly the mixtures of the island-conditionals of $\q$ (SM~\ref{ap:uniqueness}).

Substituting the universal prior into \eqref{eq:periodic_mmc} yields a prior-independent lower bound for \emph{any} periodic process with $\Cres\geq0$,
\begin{equation}\label{eq:chain}
    \C_N(p_0)\;\geq\;\MC_N(q)\;\geq\;\MC_N^*:=\MC_N(\q)\;\geq\;0,
\end{equation}
where $\C_N = \MC_N(q) + N\C(q)$ is the cost after $N$ cycles, together with the closed forms
\begin{align}
    \MC_N^* &= S(p_N)-S(p_0)+N\bigl[S(\q)-S(G\q)\bigr]\label{eq:mmc}\\
    &= N\bigl[
    \operatorname{JS}[\mathbf{p_t}] - \operatorname{JS}[\mathbf{Gp_t}
    \bigr],\label{eq:mmc_js}
\end{align}
where $\mathbf{p_t}=\{p_t : t = 0, \ldots N-1\}$, $\mathbf{Gp_t}=\{Gp_t : t = 0, \ldots N-1\}$, and $\mathrm{JS}$ is the generalized Jensen--Shannon divergence. Eq.~\eqref{eq:mmc} follows by telescoping, and \eqref{eq:mmc_js} from $\q$ being the arithmetic mean of $p_0,\ldots,p_{N-1}$ (see SM~\ref{ap:weighted_mmc_bound} and~\cite{tasnim2023stochastic}). The total cost decomposes accordingly,
\begin{equation}\label{eq:cost_qbar}
    \mathcal{C}_N(p_0) = \MC_N^* + N\,\MC_G(\q\|q) + N\,\Cres,
\end{equation}
with $\MC_G(\q\|q)$ the one-cycle MMC of initial distribution $\q$ under prior $q$.

Three properties of $\MC_N^*$ follow. (i) By the data-processing inequality for KL divergence, one can show it is \emph{strictly positive} for $N>1$ unless $G$ is logically invertible or $p_0$ is a fixed point of $G$; since it lower bounds EP, we call it the \emph{universal dissipation} (see Appendix D of Ref. \cite{gulce2024}). (ii) It is bounded above, $\MC_N^*\leq 2\log |\X|+\log N+1$ (SM~\ref{ap:upper}). (iii) Its form is unchanged if $G$ has islands, although the universal prior itself becomes a weighted mixture over islands (SM~\ref{ap:islands}). In SM~\ref{ap:periodicdeterministic}, we discuss possible caveats of deterministic dynamics over islands, which is of special interest in computational contexts.

Crucially, \eqref{eq:mmc_js} depends only on the logical dynamics $G$ and the initial distribution $p_0$; it makes no reference to a prior, i.e., to the physics. Because it follows solely from the decomposition \eqref{eq:cost}, it holds even when the dynamics implementing $G$ within a period is non-Markovian, or involves non-conservative forces, or even involves physical laws not yet discovered. This realm of applicability is far broader than the CTMC setting of essentially all prior work on periodic EP. The second law also has such universality, holding independent of the underlying physics. (That is the strength of the second law, in fact.) In this sense, our result can be viewed as the first ever strengthening of the second law. 
\newline

\emph{Asymptotics and steady states---}
If $G$ is irreducible (though not necessarily aperiodic\footnote{Periodic application of $G$ should not be confused with $G$ being a periodic stochastic matrix.}), then $\q\to\pi$, the unique fixed point of $G$, with $O(1/N)$ correction. Note that the distribution over states converges exponentially to its fixed point, but here the universal prior converges only polynomially. 

Note that in this analysis, we repeatedly applied~\cref{eq:cost} evaluated for the case where the time interval is exactly a single period, with finite length $\tau$. In contrast, whenever the fixed point is an NESS, the prior for the single infinite time interval from a finite $t$ to $t =\infty$ is not $\pi$. This discrepancy between the two priors reflects the fact that partitioning a (possibly infinite) time-interval into sub-intervals changes the prior in general. 

Such partitioning also changes the associated total MMC, never decreasing it~\cite{yadav2024mismatch}. So even though~\cref{eq:mmc_js} in the limit of $N \rightarrow \infty$ is a lower bound on the total EP over an infinite number of cycles, it is an \textit{upper} bound on the MMC evaluated for a single time interval that is the entire infinite set of cycles.

More concretely, suppose that $G$ is in fact aperiodic as well as irreducible, then
\begin{equation}\label{eq:kl_p0_pi}
    \MC_\infty^* := \lim_{N\to\infty}\MC_N^* = D(p_0\|\pi),
\end{equation}
again with $O(1/N)$ correction (details of the above statements in SM~\ref{ap:asymp}). One implication of \cref{eq:kl_p0_pi} is that the universal dissipation saturates as $N$ goes to infinity if there is a unique fixed point. 
If not, the universal dissipation can grow without bound. 
Perhaps surprisingly, this result holds even when the underlying process is not just thermal relaxation to a Boltzmann distribution $\pi$. 

If $\pi$ is an equilibrium distribution, this limit is the nonequilibrium free energy $\F(p_0)$, the maximum work extractable from $p_0$; for thermal relaxation, the MMC is in fact all of the EP. Since free-energy differences are typically macroscopically relevant quantities, the universal dissipation should likewise remain a non-negligible contribution to total entropy production in many settings. An intuitive explanation for this result can be found in (SM~\ref{ap:intuitive}). 

Recall $q$ minimizes the single-cycle total cost, while $\q$ minimizes the periodic MMC, which neglects $\Cres$. This means that for a nonequilibrium steady state (NESS) $\pi$, the actual, physical prior $q$ is generally \emph{not} $\pi$, even though $\q\to\pi$. 
Moreover, for an NESS $\Cres>0$ accumulates linearly in $N$, so the total EP grows without bound even as the universal dissipation saturates. Even if $\Cres=0$, since $q\neq \q$ in general, the $N\MC_G(\q\|q)$ term will be nonzero and growing linearly with $N$. This gives another strictly positive lower bound on EP at an NESS.

Nonetheless, Eq. \eqref{eq:kl_p0_pi} shows that any periodic process converging to $\pi$ must dissipate at least $D(p_0\|\pi)$ from the MMC component of its EP, regardless of implementation. When it is well-defined in a CTMC, this limiting quantity matches the total \emph{nonadiabatic entropy production} \cite{3flucthms_VDBESP_2010}, though our result holds more generally outside of those contexts.  

Similarly, maximizing the universal dissipation over $p_0$ gives an asymptotic worst-case value $-\max_x\log\pi(x)\geq\log |\X|$, and therefore grows logarithmically with state-space size, and is extensive. 
This means that the worst-case MMC is not infinitesimally small in the thermodynamic limit (in contrast to many of the other bounds on EP recently derived in stochastic thermodynamics).
\newline

\emph{Continuous time---}
A time-homogeneous process is a special case of a periodic process where the period is arbitrary. If described by a CTMC with rate matrix $K$, the continuous evolution is $p_t=e^{Kt}p_0$ for all $t$, so the physics describing the system are time-homogeneous. The single-cycle MMC becomes the \emph{instantaneous} MMC, which is the rate of increase of MMC at distributions $p$ and $q$ \cite{premmc2021}
\begin{equation}\label{eq:DK}
    D_K(p\|q) := -\tfrac{d}{ds}D(e^{sK} p\|e^{sK}q)\big|_{s=0}\geq0,
\end{equation}
and the accumulated MMC $\MC_\T(q)=\int_0^\T D_K(p_t\|q)\,dt$. The interpretation of the prior $q$ is now as the minimizer of the cost function rate, e.g. the EP rate. The same argument as in discrete time gives (SM~\ref{ap:continuous})
\begin{equation}\label{eq:mmc_cont}
    \q=\frac1\T\!\int_0^\T\!p_t\,dt,\qquad
    \MC_\T^*:=\MC_\T(\q)=D(p_0\|\q)-D(p_\T\|\q),
\end{equation}
with $\C_\T(p_0)=\MC_\T^*+\T D_K(\q\|q)+\T\dot{\C}(q)$, where $\dot{\C}$ is the cost rate, which is independent of $p_0$. Equality $\MC_\T^*=0$ holds only if $p_0$ is a fixed point of $K$. The same result even holds for time-inhomogeneous processes, so long as they are of the form $K(t)=w(t)K$, with a $w$-weighted time average as the universal prior (SM~\ref{ap:timedep}), e.g.\ for inhomogeneous Poisson processes. The universal dissipation bounds the thermodynamic cost of all processes mentioned above, independent of their physical implementation. 
\newline


\emph{Universal dissipation in a spin system}
We apply these results to an $n$-spin fully connected (Curie--Weiss) Ising model, $\bs\sigma\in\{-1,+1\}^{n}$, with $H(\bs\sigma)=-\frac{J}{2n}\sum_{i\neq j}\bs\sigma_i\bs\sigma_j$. We suppose that that system is coupled to two thermal reservoirs at inverse temperatures $\beta_1,\beta_2$ \cite{mamedespins2026, aguilera2023spins, martynecspins2020, tomephasetransitions}. Single-spin flips proceed through either reservoir, with total rate matrix $K=K^{(1)}+K^{(2)}$, the system obeys local detailed balance  with symmetric rates $K^{(\nu)}(\bs\sigma^{(i)}|\bs\sigma)= e^{-\beta_\nu\Delta_iH(\bs\sigma)/2}$, where $\bs\sigma^{(i)}$ is $\bs\sigma$ with spin $i$ flipped, and $\Delta_iH(\bs\sigma) = H(\bs\sigma^{(i)})-H(\bs\sigma)$\cite{vandenbroeck2015ensemble,mamedespins2026} (more details in SM~\ref{ap:spins}). Unequal temperatures break global detailed balance, so the fixed point $\pi$ is a NESS. 

\begin{figure}
    \centering
    \includegraphics[width=1\linewidth]{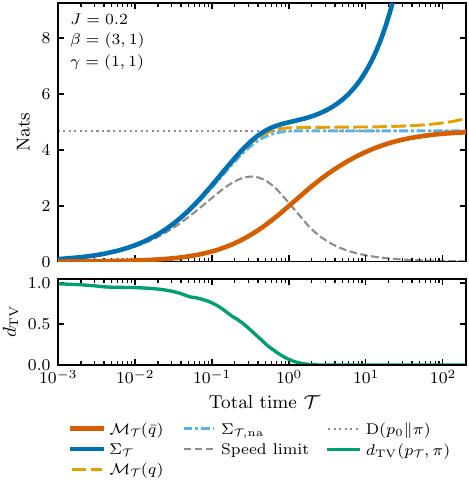}
    \caption{universal dissipation $\MC_\T^*$, MMC at the EP-optimal prior, total EP, nonadiabatic EP, and the speed-limit bound versus elapsed time for $|\Lambda|=9$ spins starting all-up, $\beta_1=3$, $\beta_2=1$, $J=0.2$. Dotted line: $D(p_0\|\pi)$. Lower panel: total variation distance $d_{\mathrm{TV}}(p_t,\pi)$ to the NESS. The universal dissipation is a significant fraction of the EP until the NESS is reached, after which residual cost dominates.}
    \label{fig:placeholder}
\end{figure}

Fig.~\ref{fig:placeholder} compares $\MC_\T(\q)$ with the total EP $\Sigma_\T$, computed from \eqref{eq:cost_form1} (SM~\ref{ap:heat_flow}). Also shown are the MMC at the true EP-optimal prior, the nonadiabatic EP $\Sigma_{\T,\mathrm{na}}=D(p_0\|\pi)-D(p_\T\|\pi)$ \cite{3flucthms_VDBESP_2010}, and the speed-limit bound $\Sigma_\T\geq 2W_1\operatorname{arctanh} \bigl(W_1/(2\T\langle\mu\rangle_\T)\bigr)$ \cite{speedlimits_ito, speedlimits_saito, speedlimits_fisher}, with $W_1=W_1(p_0,p_\T)$ the 1-Wasserstein distance and $2\langle\mu\rangle_\T$ the mean dynamical activity (SM~\ref{ap:wass1}). The universal dissipation, which depends only on $K$ and $p_0$, accounts for a substantial fraction of the total EP over an intermediate range of $T$; once the NESS is reached, the residual cost takes over. The universal dissipation converges to the nonadiabatic EP at long times, as predicted by \eqref{eq:kl_p0_pi}, and it exceeds the speed-limit bound except at short $T$, where the latter is tighter. 

It is important to note that evaluating the speed-limit theorems, non-adiabatic EP, and MMC for the true prior $q$ require far more assumptions concerning the
underlying physical dynamics than does calculating the universal dissipation, which is independent of such details. In addition, they are more difficult to compute in practice. In particular, the MMC for the true prior $\MC_T(q)$ requires solving for the minimizer of the cost function, which can be computationally expensive in large state spaces. 

\begin{figure*}
    \centering
    \includegraphics[width=1\linewidth]{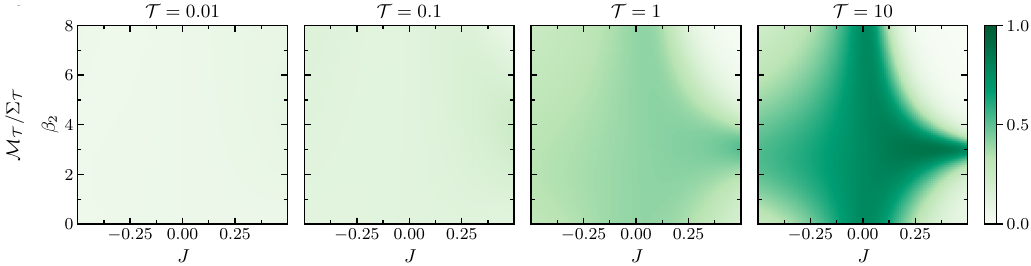}
    \caption{Ratio of universal dissipation to total EP, $\MC_\T^*/\Sigma_\T$, at several times $\T$ as a function of coupling $J\in[-0.5,0.5]$ and second-bath inverse temperature $\beta_2\in[0,8]$, with $\beta_1=3$. Over much of parameter space, particularly the antiferromagnetic regime, the universal dissipation is a significant fraction of the EP; along $\beta_2=\beta_1$ the system relaxes to equilibrium and the MMC is all of the EP.}
    \label{fig:ising_heatmap}
\end{figure*}

Fig.~\ref{fig:ising_heatmap} shows how the ratio $\MC_\T^*/\Sigma_\T$ varies with coupling strengths and relative bath temperatures ($\beta_1=3$, $\beta_2$ varies). Over a large fraction of parameter space, particularly for antiferromagnetic $J<0$, the universal dissipation is a substantial part of the EP, and in some regions nearly all of it; the sharp line at $\beta_2=\beta_1$ is where the system effectively sees a single bath and relaxes to equilibrium (SM~\ref{ap:intuitive}). This is consistent with recent results that MMC can be a significant fraction of total EP \cite{yadav_comm_2025, yadav2026entropy}, especially in scenarios of high entropy flow. 
\newline

\emph{Universal dissipation in computers---}
A natural application of the derived bound is to {uniform computers}. These are computational models like deterministic finite automata (DFA), Turing machines and space-bounded RAM machines, which have a fixed state space, process strings of any length, and implement the same physical process every clock cycle \cite{Ouldridge_2023, gulce2024, yadav2026entropy}. Their operation is a conditional map $G$ on the full state (input, internal state, memory), so $\MC_N^*$ lower bounds the EP of $N$ computational steps independently of the machine implementing it.

We consider DFAs over the alphabet $\{0,1\}$ whose set of states is $R$, whose transition rule is $\rho$ (more details in SM~\ref{ap:DFAs}), and which is initialized in $r^\varnothing\in R$. This automaton is fed i.i.d.\ symbols with $\Pr(0)=p$,
which induces a Markov chain $G$ on $R$.
We say that an input string that results in the DFA entering an accept state is a word recognized by that DFA. The set of all words that a DFA recognizes is the lagnauge recognized by the DFA. There are an infinite number of DFAs that all recognize the same language. 

Fig.~\ref{fig:DFAplot} (right) shows the minimal DFA $\mathrm{D3}$ recognizing binary integers divisible by three; the initial and accept states are both $\mathbf 0$. We also construct non-minimal DFAs $\mathrm{D5}, \mathrm{D7}, \mathrm{D9}$ that recognize the same language by unrolling the self-loops of states $\mathbf0$ and $\mathbf2$ into cycles of redundant states (SM~\ref{ap:dfa_construction}); these contain many-to-one transitions and reduce to $\mathrm{D3}$ under Myhill--Nerode minimization \cite{Nerode1958}.

\begin{figure*}
    \centering
    \includegraphics[width=\linewidth]{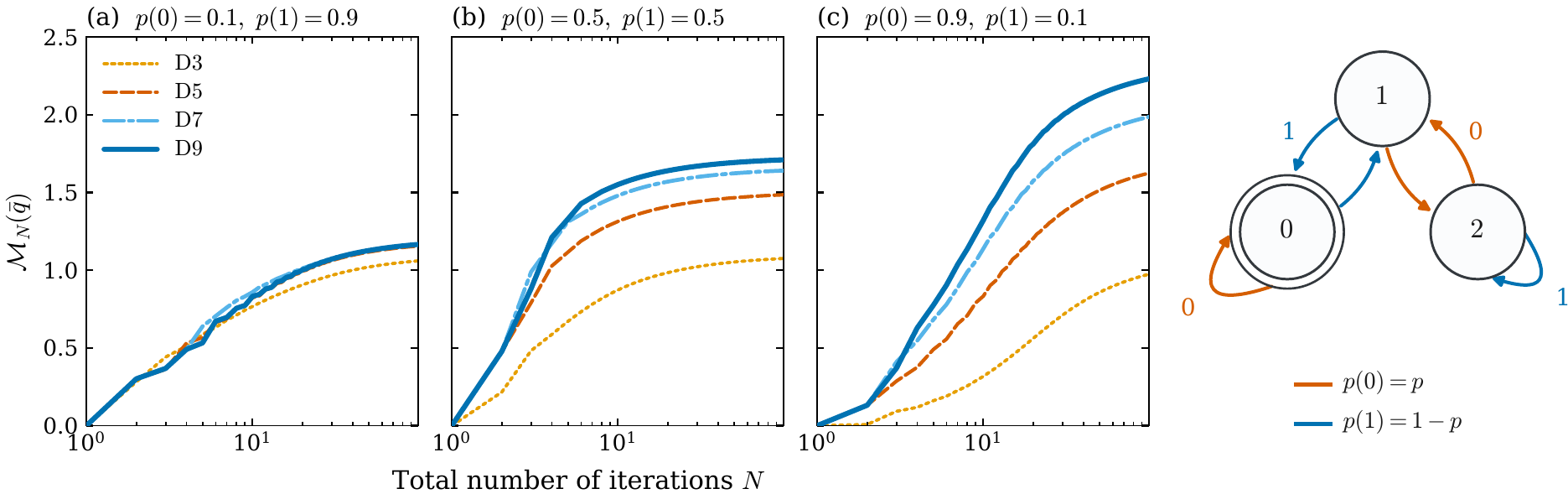}
    \caption{(Left three panels) universal dissipation $\MC_N^*$ versus number of steps $N$ for four DFAs recognizing binary integers divisible by three, for three input distributions $(p,1-p)$. (Right) The minimal DFA $\mathrm{D3}$; states in bold, initial and accept state $\mathbf 0$. At large $N$ the universal dissipation increases monotonically with the number of redundant states; at small $N$ crossovers occur.}
    \label{fig:DFAplot}
\end{figure*}

In Fig. \ref{fig:DFAplot} we plot the universal dissipation of running these DFAs for $N$ total iterations across different input distributions. For large $N$ the universal dissipation scales monotonically with the number of redundant states, and for uniform or $0$-skewed input the minimal DFA is has lower universal dissipation for all $N>1$. For $1$-skewed input, however, crossovers occur at small $N$, echoing the finding of \cite{Ouldridge_2023} that the minimal DFA need not minimize EP---though that conclusion depended on a choice of prior, whereas here it depends only on $G$, since $p_0=\delta_{r^\varnothing}$. When $G$ has a unique stationary distribution, \eqref{eq:kl_p0_pi} gives $\MC_\infty^*=-\log\pi(r^\varnothing)$, from which we prove that the long-time universal dissipation increases monotonically with the redundancy of the \emph{initial} state and is unaffected by redundancy elsewhere (SM~\ref{ap:dfa_proofs}), although the latter can matter at finite $N$. The universal dissipation thus offers a prior-free way to rank the thermodynamic cost of DFAs recognizing the same language, and to identify which redundancies raise or lower it. A visualization of the universal prior trajectory on the simplex is given in the SM.
\newline

\emph{Discussion---}
In this letter, we derived a strictly positive lower bound on the EP---and any cost of the form \eqref{eq:cost_form1}---of any process that evolves an initial distribution $p_0$ by repeated application of a map $G$. It depends only on $G$, $p_0$ and $N$, requires neither CTMC dynamics within a period nor (local) detailed balance, extends to time-homogeneous continuous-time processes. This therefore constitutes a strict strengthening of the second law for periodic processes. The prior that achieves it is the time average of the trajectory, and the bound itself is a drop in Jensen--Shannon divergence between the trajectory and its one-step shift. When $G$ has a unique fixed point $\pi$, the bound converges to $D(p_0\|\pi)$: the nonequilibrium free energy if $\pi$ is equilibrium, and otherwise the dissipation the process would incur were $\pi$ equilibrium. We emphasize that the universal dissipation is a macroscopically relevant quantity, and show that, in the worst case (over inputs), it scales extensively with system size.

For spins between two baths the universal dissipation is a non-negligible, and in several regimes dominant, fraction of the EP; it matches the nonadiabatic EP at long times and outperforms speed-limit bounds outside the short-time regime, suggesting it is a more general form of the dissipation incurred whenever a distribution is evolved forward periodically by a stochastic map. For DFAs, the universal dissipation ranks machines recognizing the same language, and we proved that redundancy in the initial state monotonically increases the long-time cost.

As a next step, studying dynamics richer than relaxation to an NESS and different connection topologies in the spin system would clarify how much of total EP the universal dissipation captures more broadly, and characterizing which non-minimal DFA structures \emph{lower} it remains open. Communication networks---encoders, decoders, modulation---are periodic in the same sense and have been studied with illustrative priors \cite{yadav_comm_2025, message_passing_2026} which our results remove the need for, and since MMC exists for continuous state spaces \cite{premmc2021}, could extend to optical channels. Finally, the results apply to any repeated process whose cost-minimizing prior is the same at each application, including repeated contact with a thermal bath and certain models of chemical reaction networks.

Finally, the strictly positive lower bound on EP given by the universal dissipation in \cref{eq:mmc_js} is easily computable from experimental data. Given that the process is periodic, one does not actually need the per-period map $G$ to compute the JS divergence, just the marginals $p_t$ at each period. This can be obtained as a histogram of various trajectories of the same periodic process, computed at each period.

\begin{acknowledgments}
We acknowledge the support of CSH, especially Carrie Cowan, during the DIS masterclass on stochastic thermodynamics of computation. B.A. and D.H.W. are supported by a research grant (VIL73405) FOCUSED from Villum Fonden. J.K. acknowledges support from Austrian Science Fund (FWF), Grant-DOI: 10.55776/PAT1151725. Claude Opus was used to generate consistent coloring and style across figures and to provide feedback on writing and mathematical arguments.
\end{acknowledgments}

\bibliography{references}

\section{End Matter}
\subsection{Universal prior on the simplex}\label{ap:simplex}
Fig.~\ref{fig:simplex_div3} fixes $p_0$ and shows the universal prior trajectory on the 2-simplex as a function of $N$ for $\mathrm{D3}$. The universal prior approaches the uniform stationary distribution of $G$, while the color shows the (log-scale) periodic MMC $\MC_N(q)$ for any other prior $q\neq\q$. The MMC for all $q$ grows with $N$, priors near $\q$ incur the least cost, and priors near the boundary of the simplex the most.

\begin{figure}
    \centering
    \includegraphics[width=1\linewidth]{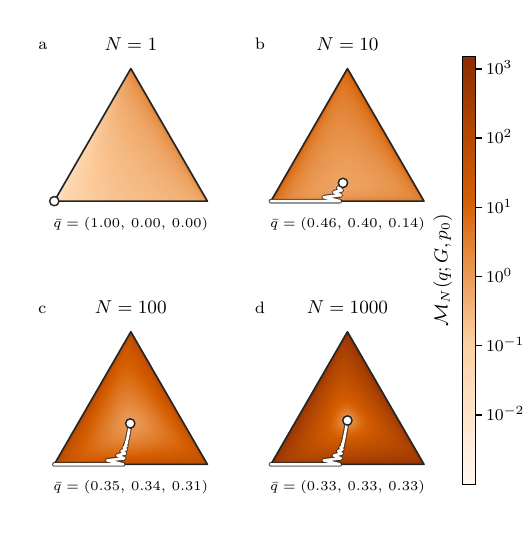}
    \caption{Heatmap of $\MC_N(q;p_0)$ over the simplex for several $N$ for the $\mathrm{D3}$ DFA with $p=0.1$. White line: trajectory of the universal prior up to the $N$-th iteration (white dot). At $N=1$ the universal prior is $p_0$.}
    \label{fig:simplex_div3}
\end{figure}

\subsection{Information Theory Background}
Throughout this paper, several concepts from information theory are used. Two of the most important are the Kullback-Leibler Divergence and the (generalized) Jensen Shannon (JS) Divergence. Both of these quantities are ways of comparing probability distributions, though neither are metrics in the mathematical sense. The square root of the JS divergence between two distributions is a metric, however. 
\subsection{Kullback-Leibler Divergence}
The Kullback-Leibler (KL) Divergence (or relative entropy) between two probability distributions $p$ and $q$ over a sample space $\mathcal{X}$ is given by
\begin{equation}
    D(p\|q) = \sum_{x\in \mathcal{X}} p(x)\log\frac{p(x)}{q(x)},
\end{equation}
where $0\log 0=0$ and $D(p\|q)=0$ if and only if $p=q$. This measures cost of approximating $p$ with $q$. The KL divergence is not symmetric, nor does it satisfy the triangle inequality \cite{coverandthomas}. 
\subsection{Jensen-Shannon Divergence}
The Jensen Shannon divergence is a symmetrized KL divergence. Define the mixture distribution $m := \frac{p+q}{2}$. Then
\begin{equation}
    \mathrm{JS}(p\|q) = \frac{1}{2}D(p\|m)+\frac{1}{2}D(q\|m).
\end{equation}
Interestingly, $\sqrt{\mathrm{JS}(p\|q)}$ is a metric in the mathematical sense. One can also define the generalized JS divergence to compare multiple probability distributions at once:
\begin{equation}
    \mathrm{JS}_{\alpha_1,\ldots,\alpha_n}(p_1, \ldots,p_n) = \sum_i \alpha_i D(p_i\|m),
\end{equation}
where $m:= \sum_{i=1}^n\alpha_i p_i$ and $\alpha_i$ are weights that are selected for the probability distributions. The Jensen Shannon divergence appears a lot in nonequilibrium thermodynamics, an example being its relationship to the Thermodynamic Length \cite{thermolengthcrooks}. 

\clearpage
\onecolumngrid
\setcounter{section}{0}
\setcounter{secnumdepth}{3}
\setcounter{equation}{0}
\setcounter{table}{0}
\renewcommand{\thesection}{S\Roman{section}}
\renewcommand{\theequation}{S\arabic{equation}}
\renewcommand{\thetable}{S\Roman{table}}
\renewcommand{\theHsection}{supp.\arabic{section}}
\renewcommand{\theHequation}{supp.\arabic{equation}}
\renewcommand{\theHtable}{supp.\arabic{table}}

\begin{center}
{\large\bfseries Supplemental Material for ``Strengthened Second Law for Periodic Processes''}\par
\vspace{0.75em}
Jake Schaefer, Ben Ansbacher, Jan Korbel, \& David H. Wolpert
\end{center}
\vspace{0.75em}


\section{Terminology}
Here, we introduce some terminology. In the following, $\mathcal{X}$ and $\mathcal{Y}$ denote finite state spaces with $\Delta_\mathcal{X}$ the space of distributions on $\mathcal{X}$. Consider a conditional distribution $G(y|x)$, specifying the probability of $y\in \mathcal{Y}$ given $x\in\mathcal{X}$. Given an initial distribution $p_0\in\Delta_\mathcal{X}$, we write $p_t:=G^tp_0$ as the distribution after $t$ applications of conditional distribution $G$.

\section{Islands}\label{ap:islands}
In study \cite{kolchinsky_circuits}, authors noticed that in some cases of $G$ the distribution $q$ in \eqref{eq:cost} that minimizes the thermodynamic cost is not unique. This motivates discussing an important mathematical structure called the \emph{islands} of $G$. For states $x,x'\in\mathcal{X}$ and $y\in\mathcal{Y}$, consider the relation
\begin{equation}
    x\sim x' \iff \exists y: G(y|x)>0, \quad G(y|x')>0. \label{eq:equivrelation}
\end{equation}
Broadly speaking, $x\sim x'$ if there is a non-zero probability of transitioning to the state $y$ from both $x$ and $x'$ under $G$. An \emph{island} of $G$ is defined to be any connected subset of $\mathcal{X}$ given by the transitive closure of this relation. The set of islands of $G$ forms a partition of $\mathcal{X}$, which we denote as $L(G)$ \cite{kolchinsky_circuits}. We can also consider the islands of $G$ restricted to any subset of the state space $\mathcal Z \subseteq \mathcal {X}$. In this case, $L_\mathcal{Z}(G)$ denotes the partition of $\mathcal Z$ generated by the transitive closure of \eqref{eq:equivrelation} for $x,x'\in \mathcal Z$. 

For a distribution $p$ over $\mathcal{X}$, we use $\mathrm{supp}\; p:=\{x\in \mathcal X \mid p(x) >0\}$ to indicate the positive support of $p$ and write $p(\mathcal{Z}) = \sum_{x\in \mathcal{Z}} p(x)$ to mean the total probability in the subset $\mathcal Z$. 
As such, we construct the conditional probability of the state $x$ within an island $\ell\in L_\mathcal{Z}(G)$ as
\begin{equation}\label{eq:island_conditional}
    p(x|\ell) := 
    \begin{cases}
        \dfrac{p(x)}{p(\ell)} \quad &\text{if } x\in \ell \text{ and } \ell\subseteq\text{supp } p \\
        0 \quad &\text{otherwise}
    \end{cases}.
\end{equation}
As a shorthand, we will write $p^\ell$ for the transition matrix $p(\cdot|\ell)$. 
When there are multiple islands, the cost function \eqref{eq:cost_form1}, and therefore the prior, decomposes into a linear combination over the islands, spelled out in full detail in Theorem \ref{thm:islands} in SM \ref{ap:mmc_islands} \cite{kolchinsky_circuits}. 

Intuitively, these islands are subsystems that are both logically and thermodynamically isolated from each other while implementing $G$. The impact of islands on MMC and entropy production has been studied at length in computational systems \cite{Ouldridge_2023,kolchinsky_circuits}. In particular, islands can be used to study the mismatch cost associated with the \emph{modularity} of a computational system.
\newline

In the case where we have a collection of islands $\ell\in L_\Ze(G)$, the \emph{universal prior} is given by
\begin{equation}\label{eq:qstarislands}
    \q_{L}(x) := \sum_\ell c_\ell \q(x|\ell) \in \argmin_{q\in\Delta_\mathcal{X}} \MC_N(q),
\end{equation}
where $c\in\Delta_{L_\Ze(G)}$ and $\q^\ell\equiv\q(\cdot|\ell)$ are the priors $\q$ conditional on each island $\ell$ given by \eqref{eq:island_conditional}. If $\mathcal{Z}$ is a set on which $\q$ has full support, then 
\begin{equation}
    \argmin_{q\in\Delta_\mathcal{X}}\MC_N(q) = \left\{\sum_{\ell\in L_\Ze(G)} c_\ell \q^\ell:c\in\Delta_{L_\Ze(G)} \right\},
\end{equation}
A proof is provided in SM \ref{ap:uniqueness}. We show that the minimal periodic mismatch cost in this case is identical to \eqref{eq:mmc_js}:
\begin{align}
    \MC_N^*
    &=S(p_N)-S(p_0)+N\!\left[S(\q)-S(G\q)\right]\\
    &=
    N[\operatorname{JS}(p_0,...,p_{N-1})
    -\operatorname{JS}(p_1,...,p_N)]
    \label{eq:mmc_js_islands}.
\end{align}
This demonstrates the minimal periodic MMC is independent of the number of islands, whereas the universal prior is not. The proof is in SM \ref{ap:weighted_mmc_bound}.

\section{Theorem on Mismatch Cost with islands}\label{ap:mmc_islands}
\begin{thm}\label{thm:islands}\cite{kolchinsky_circuits}
 Consider any function $\C:\Delta_\mathcal{X} \to \mathbb R$ of the form $\Ce(p):=S(Gp)-S(p) + F(p),$ where $G(y|x)$ is some conditional distribution of $y\in \mathcal{Y}$ given $x\in \X$ and $F$ is a linear functional $F(p) = \sum_{x\in \X} p(x)f(x)$. Let $\Ze$ be any subset of $\X$ such that $f(x)<\infty$ for $x\in\Ze$ and let $q\in \Delta_\Ze$ be any distribution that obeys
\begin{equation}
    q^\ell \in \argmin_{r\,:\,\operatorname{supp} r\subseteq \ell} \Ce (r) \quad \forall \ell\in L_\Ze(G).
\end{equation}
Then each $q^\ell$ is unique and for any $p$ with $\operatorname{supp} p\subseteq \Ze$, 
\begin{equation}
    \Ce(p) = \MC(p) + \sum_{\ell\in L_\Ze(G)} p(\ell) \Ce(q^\ell),
\end{equation}
where the second term is the residual cost. In addition, we have
\begin{equation}
    \Ce(p) = \sum_{\ell\in L_\Ze(G)}p(\ell)\Ce(p^\ell),
\end{equation}
meaning the cost function decomposes into a weighted sum over the cost function of each island. Therefore, the mismatch cost does the same:
\begin{equation}\label{eq:mmc_island_decomp}
    \MC(q) = \sum_{\ell\in L_{\mathcal{Z}}(G)} p(\ell)\MC(q^\ell).
\end{equation}
\end{thm}

\noindent\textbf{Corollary.} The periodic MMC decomposes as
\begin{equation}
    \MC_N(q) = \sum_{\ell\in L_{\mathcal{Z}}(G)} \MC_N^\ell(q^\ell),
\end{equation}
where
\begin{equation}\label{eq:mmc_island_kl}
    \MC_N^\ell(q^\ell) = \sum_{t=0}^{N-1} p_t(\ell)
      \Bigl[D(p_t^\ell\|q^\ell) - D(Gp_t^\ell\|Gq^\ell)\Bigr].
\end{equation}

\section{General proof of weighted optimal prior}
\label{ap:weighted_pop}
\noindent\textbf{Proposition.} Assume that the mismatch cost is weighted by some non-negative time-dependent weight function $w_t$, so the periodic mismatch cost takes the form
\begin{equation}\label{eq:weighted_mmc_kl}
    \MC_{w,N}(q) = \sum_{t=0}^{N-1} w_t[D(p_t\|q)-D(Gp_t\|Gq)].
\end{equation}
We denote the value that $\MC_{w,N}$ takes at its minimum as 
\begin{equation}
    \MC^*_{w,N} := \min_{q\in\Delta_\mathcal{X}}\MC_{w,N}(q).
\end{equation}
Define the weighted trajectory-averaged distribution,
\begin{equation}\label{eq:weighted_pop}
    \q_w := \frac{1}{W} \sum_{t=0}^{N-1}w_tp_t,\quad\text{with } W:=\sum_{t=0}^{N-1} w_t>0.
\end{equation}
Then, this is a minimizer of the periodic mismatch cost:
\begin{equation}
    \q_w \in \argmin_{q\in\Delta_\mathcal{X}} \MC_{w,N}(q).
\end{equation}
Note that this proposition applies directly to a single-island mismatch cost and will subsequently be applied to each separate island.

\begin{proof}
Note that in general, for finite state spaces,
\begin{equation}
    D(p\|u)- D(p\|v) =p^\top\log\frac{v}{u}.
\end{equation}
Thus, grouping like terms in this manner, the mismatch cost difference simplifies to
\begin{align}
    \Delta\MC_{w,N}&:=\MC_{w,N}(q) - \MC_{w,N}(\q_w)\nonumber\\
    &= \sum_t w_t\left[p_t^\top\log\frac{\q_w}{q} - (Gp_t)^\top\log\frac{G\q_w}{Gq}\right]\nonumber\\
    &= W [D(\q_w\|q) - D(G\q_w\|Gq)],
\end{align}
where the last equality follows from using $\sum_tw_tp_t=W\q_w$ and linearity of $G$. Thus, by the data-processing inequality for KL divergence,
\begin{equation}
    \MC_{w,N}(q)\geq\MC_{w,N}(\q_w)
\end{equation}
for all $q\in\Delta_{\mathcal{X}}$. Therefore, $\q_w$ is a globally valid minimum.
\end{proof}

\noindent\textbf{Corollary.} Define the trajectory-averaged distribution:
\begin{equation}
    \q:=\frac{1}{N}\sum_{t=0}^{N-1}p_t.
\end{equation}
This is the global minimizer of the mismatch cost $\MC_N(q)$ over all possible priors $q\in\Delta_{\mathcal{X}}$ for a single island, which is true from the above proof with $w_t=1$.

\section{Uniqueness of prior}
\label{ap:uniqueness}
\noindent\textbf{Proposition.} Any prior $q\in\Delta_\mathcal{X}$ whose conditional distribution on island $\ell\in L_{\mathcal{Z}}(G)$ is exactly $q(x|\ell) = \q(x|\ell)$ is a minimizer of the periodic MMC:
\begin{equation}\label{eq:mmc_island_minimizing_set}
    \left\{\sum_{\ell\in L_\Ze(G)} c_\ell \q^\ell:c\in\Delta_{L_\Ze(G)} \right\}\subseteq\argmin_{q\in\Delta_\mathcal{X}}\MC_N(q).
\end{equation}
In addition, if $\q$ has full support over $\mathcal{Z}$, then the converse holds over $\Delta_\mathcal{Z}$: the prior $q\in\Delta_\mathcal{Z}$ which minimizes $\MC_N$ is unique up to a piecewise constant function on islands. In particular,
if $\mathcal{Z}=\mathcal{X}$, then this characterizes all global minimizers.

\begin{proof}
Let $X$ and $Y$ be random variables, with $X\to Y$ governed by forward process $G(y|x)$. By Theorem~\ref{thm:islands}, the MMC decomposes over islands $\ell \in L_{\mathcal{Z}}(G)$ as
\begin{equation}
    \MC_N(q) = \sum_{\ell\in L_{\mathcal{Z}}(G)} \MC_N^\ell(q^\ell).
\end{equation}
SM~\ref{ap:weighted_pop} applies to the $\ell$-th term with weight function $w_t=p_t(\ell)$ and distributions $p_t\mapsto p_t^\ell$ and $q\mapsto q^\ell$. Thus, a minimizer of each $\MC_N^\ell$ exists and is exactly the conditional distribution of $\q$ on $\ell$:
\begin{equation}
    \q^\ell = \frac{1}{W_\ell}\sum_t p_t(\ell)p_t^\ell = \frac{\q(x)}{\q(\ell)}\mathbbm{1}_{\ell \cap \text{supp }\q}.
\end{equation}
Since $\q^\ell$ is a minimizer, we know that 
\begin{equation}\label{eq:mmc_diff_islandwise}
    \MC_N^\ell(q^\ell) - \MC_N^\ell(\q^\ell) \geq 0
\end{equation}
for all $q\in\Delta_\mathcal{X}$. From the above decomposition, the value of $\MC_N(q)$ depends only on the conditionals $q^\ell$ and not on the masses $q(\ell)$ assigned to each island. Thus, any distribution $q$ with conditional distribution such that $q^\ell = \q^\ell$ on each island equals is a minimizer of the total MMC:
\begin{equation}
    \MC_N(q) = \sum_\ell \MC_N^\ell(\q^\ell) = \MC_N^*,
\end{equation}
proving the inclusion \eqref{eq:mmc_island_minimizing_set}, where $c_\ell=q(\ell)$.

We now prove the converse under the full-support assumption and assuming the optimization domain is over $\Delta_\mathcal{Z}$. From SM. \ref{ap:weighted_pop}, any minimizer $q^\ell$ of $\MC_N^\ell$ satisfies
\begin{align}\label{eq:min_kl_diff}
    0 
    &= D(\q^\ell\|q^\ell) - D(G\q^\ell\|Gq^\ell) \\
    &= \sum_y (G\q^\ell)_y D(\bar{Q}_{X|y}\|Q_{X|y}),
\end{align}
where the last equality is the chain-rule for KL divergence, with $\bar{Q}_{X|Y}$ and $Q_{X|Y}$ being the backward processes for $G$ of $\q^\ell$ and $q^\ell$, respectively. 
We posit that all minimizers are equivalent under the relation that they have the same Bayesian inverse of $G$.
Thus, the equality \eqref{eq:min_kl_diff} holds if, for all $y$ on the support of $G\q^\ell$, the backward processes, or Bayesian inverses, are equivalent: $\bar{Q}_{X|Y=y}=Q_{X|Y=y}$. A single matrix should lead to the perfect retrodiction in time of both $q^\ell$ and $\q^\ell$. Thus, from Bayes' rule,
\begin{equation}
    \bar{Q}(x|y) 
    = \frac{G(y|x)\q^\ell(x)}{(G\q^\ell)_y},
\end{equation}
we have equality if
\begin{equation}
    \frac{q^\ell(x)}{\q^\ell(x)} = \frac{(Gq^\ell)_y}{(G\q^\ell)_y}.
\end{equation}
Now, since the right hand side consists of marginal distributions of $y$, if we have two inputs $x$ and $x'$ that can give the same output $y$, so $G(y|x)>0$ and $G(y|x')>0$, applying the equation transitively gives
\begin{equation}
    \frac{q^\ell(x)}{\q^\ell(x)}=\frac{q^\ell(x')}{\q^\ell(x')}.
\end{equation}
Thus, taking the transitive closure, we get that on each island $\ell$, any other minimizer $q^\ell$ must satisfy
\begin{equation}
    q^\ell(x) = r_\ell \q^\ell(x).
\end{equation}
Both are normalized on $\ell$, so $r_\ell=1$. Hence, 
\begin{equation}
    q^\ell(x) = \q^\ell(x).
\end{equation}
The inequality \eqref{eq:mmc_diff_islandwise} together with positivity of the MMC implies that equality of the sum with zero requires equality on each island.
Therefore, inclusion within the minimizing set of distributions shown in \eqref{eq:mmc_island_minimizing_set} is actually an equality, as desired.
\end{proof}

\noindent\textbf{Corollary.} If there is only a single island and $\q$ has full support on it, it is clear that $\q$ is the unique global minimizer.

\section{Weighted periodic MMC bound}
\label{ap:weighted_mmc_bound}
Here, we prove two general equalities of the weighted MMC evaluated at the weighted periodic optimal prior $\q_w$ of equation \eqref{eq:weighted_pop}.

\begin{proof}
Recall that the weighted MMC can be written as a sum of KL divergences as in \eqref{eq:weighted_mmc_kl}. Evaluating at $\q_w$ and expanding the KL-divergence into the form $D(u \| v) = -S(u) - u^\top \log v$, we split the objective into entropy terms and cross-entropy terms:
\begin{align}
    \MC_{w,N}^*
    &= \sum_t w_t[S(Gp_t)-S(p_t)] \nonumber\\
    &+ \sum_t w_t[(Gp_t)^\top\log G\q_w - p_t^\top\log \q_w].
\end{align}
Then, by linearity and recalling the definition of $\q$, we have that the second sum reduces to a difference of entropies:
\begin{align}\label{eq:mmc_weighted_entropy_diff}
    \MC_{w,N}^*
    &= \sum_t w_t[S(Gp_t)-S(p_t)] \nonumber\\
    &+ W[S(\q_w) - S(G\q_w)].
\end{align}

Now, note that the Jensen–Shannon divergence for distributions $p_0,\ldots, p_{N-1}$ with weights $\alpha_0,\dots, \alpha_{N-1}$ is defined as 
\begin{equation}\label{eq:jsdivdefn}
    \operatorname{JS}_\alpha((p_t)_{t=0}^{N-1}) = \sum_{t=0}^{N-1} \alpha_t D\left(p_t\middle\|\sum_{s=0}^{N-1} \alpha_s p_s\right).
\end{equation}
Thus, identifying $\alpha_t = w_t/W$ and using the definition of the periodic-optimal prior $\q$, we get that 
\begin{align}\label{eq:mmc_weighted_JS}
    &\MC_{w,N}^*
    = W[\operatorname{JS}_\alpha((p_t)_{t=0}^{N-1}) - \operatorname{JS}_\alpha((Gp_t)_{t=0}^{N-1})],
\end{align}
and this concludes the proof.
\end{proof}
This result can also be derived by considering MMC of a periodic process parameterized by $t$, by analyzing  uncertainty about the initial distribution $p_0$ \cite{tasnim2023stochastic}. However, the analysis in that paper does not solve for the minimizing prior, nor does it derive any extensions or provide interpretations. 
\newline

\noindent\textbf{Corollary.} The periodic MMC with $w_t=1$ obeys the equations \eqref{eq:mmc} and \eqref{eq:mmc_js} at its minimal value $\MC_N^*:=\MC_N(\q)$. Indeed, this is true regardless of the number of islands.
\begin{proof}
We have already shown in SM. \ref{ap:uniqueness} that any distribution satisfying $q^\ell =\q^\ell$ is indeed the optimal prior for each $\MC_N^\ell$ and thus $\MC_N$. In particular, the trajectory averaged universal prior $\q$ trivially satisfies this and is a minimizer for any number of islands, so 
\begin{equation}
    \MC_N^* = \MC_N(\q)
\end{equation}
always. Since $w_t=1$, we have that $W=N$. Also, the entropies in the first term of \ref{eq:mmc_weighted_entropy_diff} telescope, which gives
\begin{equation}
    \MC_N^* = [S(p_N)-S(p_0)] + N[S(\q) - S(G\q)].
\end{equation}
Furthermore, we have from $\alpha_n=1/N$ that this is also equal to the following difference of Jensen-Shannon divergences for equally-weighted distributions:
\begin{equation}
    \MC_N^* = N[\operatorname{JS}((p_n)_{n=0}^{N-1}) - \operatorname{JS}((Gp_n)_{n=0}^{N-1})].
\end{equation}
Thus, these two equalities hold for the MMC at $\q$ independently of the island-structure.
\end{proof}

\section{Comment on deterministic periodic dynamics}\label{ap:periodicdeterministic}
In computational systems the map $G$ is often deterministic, meaning
$G(y|x)=\delta_{g(x),y}$ for some function $g$ that may be many-to-one.
Consider initializing such a system in a delta function, $p_0(y)=\delta_{x_0,y}$.
Then $p_t(y)=\delta_{g^t(x_0),y}$ for all $t$, so $S(p_t)=0$ for all $t$ and
the first term in Eq.~(7) vanishes:
\begin{equation}
  \MC_N^* = N\bigl[S(\q)-S(G\q)\bigr],
  \qquad
  \q(y)=\frac{1}{N}\sum_{t=0}^{N-1}\delta_{g^t(x_0),y},
  \quad
  (G\q)(y)=\frac{1}{N}\sum_{t=1}^{N}\delta_{g^t(x_0),y}.
\end{equation}
Equivalently, $\MC_N^*$ is $N$ times the difference between the Shannon entropy of
the empirical visit distribution of $x_0,\dots,x_{N-1}$ and that of $x_1,\dots,x_N$.

Suppose further that a particular trajectory of such a system never visits any state more than once
before it halts, i.e. \ $g^t(x_0)\neq g^k(x_0)$
for all $0\le k<t\le N$. Then $\q$ and $G\q$ are both uniform on $N$ distinct
states, so $S(\q)=S(G\q)=\log N$ and
\begin{equation}
  \MC_N^*=0 .
\end{equation}
Note that $\MC_N^*=0$ means only that this prior-independent bound is trivial for this
$p_0$; it does not imply that the entropy production vanishes, since the total cost also
contains the residual cost and the mismatch term $N\,\MC_G(\q\|q)$ in \cref{eq:cost_qbar}.

Now consider a system following such a deterministic dynamics where - as is essentially always the case
in the real world - the support of the initial distribution includes (very) many states. Invariably
in such a situation, there will be states that have nonzero probability mass at multiple times.
In particular, this is true for essentially all real-world ``uniform’’ computers, in which the
halting time varies depending on the precise input.

In all such systems, assuming the deterministic dynamics includes many-to-one maps (i.e.,
isn’t just a permutation of the states), our bound will be strictly positive the bound is strictly positive when some island $\ell$ of $G$ carries probability mass at two or more
iterations with conditional distributions $p_t^{\ell}$ that differ across those
iterations, because the data-processing inequality is then strict. For example, for
$g(1)=g(2)=1$ and $p_0=(\tfrac12,\tfrac12)$ one finds $\MC_2^*=2H(\tfrac14)-\log 2\approx 0.43$
nats, whereas the same $g$ with $p_0=\delta_1$ gives $\MC_N^*=0$.

\section{Continuous-time MMC}\label{ap:continuous}
Consider a continuous process with a time-homogeneous master equation
\begin{equation}
    \frac{d}{dt} p(t) = Kp(t),\quad p(0)=p_0,
\end{equation}
where $K$ is a time-independent rate matrix. Thus, an initial distribution $p_0$ evolves according to
\begin{equation}
    p_t = e^{Kt}p_0.
\end{equation}
Therefore, to relate this to the discrete-time mismatch cost, define for discrete time intervals $[0,\eps]$ the transition matrix $G_\eps:=e^{\eps K}$. The MMC for a single step over the interval $[0,\eps]$ is:
\begin{align}
    \MC_\eps(p_0) 
    &= D(p_0\|q) - D(G_\eps p_0\|G_\eps q) \nonumber\\
    &=D(p_0\|q) - D(e^{\eps K} p_0\|e^{\eps K} q).
\end{align}
This is a KL divergence contraction. Now, define the continuous KL contraction by taking the limit of these discrete-time contractions:
\begin{align}
    D_K(p\|q)
    &=\lim_{\eps \downarrow 0} \frac{D(p\|q) - D(G_\eps p\|G_\eps q)}{\eps} \label{eq:cont_mmc_limit}\\
    &=\left.-\frac{d}{d\eps}D(G_\eps p\| G_\eps q)\right|_{\eps=0}.
\end{align}
We denote this as the instantaneous mismatch cost, which can be shown to satisfy the instantaneous relation in regards to the entropy production when $q$ is the optimal prior for the rate of entropy production\cite{premmc2021}:
\begin{equation}
    D_K(p\|q) =\dot \Sigma(p)-\dot\Sigma (q).
\end{equation}
The data-processing inequality for $e^{\eps K}$ implies that the numerator of \eqref{eq:cont_mmc_limit} is non-positive for all $\eps\geq 0$, so
\begin{equation}
    D_K(p\|q)\geq 0.
\end{equation}
Then, for over interval of length $T$, the continuous-time periodic mismatch cost can be defined via a limit of Riemann sums as 
\begin{equation}
    \lim_{N\to\infty}\sum_{n=0}^{N-1} D_K(p_{n\Delta t}\|q)\Delta t = \int_0^T D_K(p_t\|q)dt,
\end{equation}
where $\Delta t = T/N$ and the integrand is for $p=e^{tK}$ with fixed $t$. 

Now, note that the time derivative of the KL divergence for time-dependent $p(t)$ and $q(t)$ is 
\begin{align}
    \frac d{dt} D(p\| q) &= \dot p^\top\log \frac pq + p^\top\left(\frac{\dot p}{p} - \frac{\dot q}{q}\right) \\
    &= \dot p^\top\left(\log \frac pq + \mathbf{1}\right) - p^\top \frac{\dot q}{q}
\end{align}
Thus, writing $p_t=e^{tK}p$ and $q_t=e^{tK}q$, plugging into the previous equation, and evaluating at $t=0$, we have the identity
\begin{align}
    \left.\frac{d}{dt}D(p_t\| q_t)\right|_{t=0}
    &= (Kp)^\top\log \frac pq
    - p^\top \frac{Kq}q.
    \label{eq:KL_time_deriv}
\end{align}
This is due to the fact that $\dot p_t = Kp_t$, and $\mathbf{1}^\top K=0$ vanishes since $K$ is a rate matrix. Thus, the instantaneous mismatch cost is
\begin{align}
\label{eq:cont_KL_expansion}
    \MC(q) = D_K(p\| q) &= -(Kp)^\top\log \frac pq +p^\top \frac{Kq}q
\end{align}

Now, we must prove that the minimizer of this is the trajectory-average
\begin{equation}
    \q = \frac1T\int_0^T p_t dt.
\end{equation}
To show this is a global minimum, again first expand using the previous identity for $D_K(p\|q)$ and grouping like terms,
\begin{align}
    \Delta\MC_T 
    = \int_0^T (Kp_t)^\top\log \frac q\q + p_t^\top \left(\frac{Kq}q - \frac{K\q}\q \right)dt.
\end{align}
Note that the two terms have time-dependence in only $Kp_t$ and $p_t$. Thus, by linearity, the first term is
\begin{align}
    \int_0^T (Kp_t)^\top\log \frac q\q dt &= -T (K\q)^\top \log \frac \q q.
\end{align}
Similarly, the second term becomes
\begin{align}\label{eq:cont_int_diff}
    \int_0^T p_t^\top\left(\frac{Kq}q - \frac{K\q}\q \right)dt
    = T \q^\top \frac{Kq}q,
\end{align}
where the last equality is due to the fact that $K$ is a rate matrix and $ \mathbf{1}^TK\q = 0$. Thus, combining the two, the difference is
\begin{equation}
    \Delta\MC_T =  
    T\left[ \q^\top \frac{Kq}q - (K\q)^\top \log \frac \q q \right],
\end{equation}
which we saw using \eqref{eq:cont_KL_expansion} to be exactly
\begin{equation}
    \MC_T(q)-\MC_T(\q) = T D_K(\q\| q)\geq 0,
\end{equation}
where the inequality follows from the DPI. Thus, this is the minimizer. 
Now, to find the mismatch cost at $\q$, we use \eqref{eq:cont_KL_expansion}, but note that the second term in that equation vanishes when integrated with respect to $t$ as it did in \eqref{eq:cont_int_diff}. Also, since $\q$ is constant over time, we can note using \eqref{eq:KL_time_deriv} that the first term is just the time derivative of $D(p_t\|\q)$ for constant $\q$:
\begin{align}
    \MC_T(\bar q) 
    &= -\int_0^T\frac d{dt}D(p_t\|\q) dt \\
    &= D(p_0\|\q)-D(p_T\|\q).
\end{align}
This version has a stark resemblance to the result from the discrete time case.  Alternatively, using that $p_T-p_0 = TK\q$, 
\begin{equation}
    \MC_T(\bar q) = S(p_T)-S(p_0)+T(K\q)^\top \log\q.
\end{equation}

\subsection{Time-dependent rate matrices}\label{ap:timedep}
Let the state space $\X$ be finite and let $K:[0,T]\to\mathbb{R}^{|\mathcal{X}|\times |\mathcal{X}|}$ be a continuously differentiable family of rate matrices. Consider the time-inhomogeneous master equation
\begin{equation}
    \frac d{dt}p_t = K_tp_t,\quad p_s=p,
\end{equation}
which has solutions $p_t = U(t, s)p$, where the evolution operator $U(t,s)$ is the unique solution of
\begin{equation}
    \partial_t U(t,s) = K_t U(t, s),\quad U(s,s)=\mathbb{I}.
\end{equation}
Evaluating the left-hand side at $t=s$, we have
\begin{equation}
    \partial_t U(t,s)|_{t=s} = K_s U(s,s) = K_s,
\end{equation}
so Taylor expansion in the first argument gives
\begin{equation}
    U(s+\eps,s)=\mathbb{I}+\eps K_s + O(\eps^2) = e^{\eps K_s}+O(\eps^2).
\end{equation}
Thus, the solutions are exponential to first order. 
Since the KL-divergence is continuously differentiable on the interior of the probability simplex, the second order terms $O(\eps^2)$ do not impact the first-order limit. Thus, assuming $p$ and $q$ have full support, and taking our definition of the continuous KL contraction,
\begin{align}
     D_{K_s}(p\|q)
    &=\lim_{\eps \downarrow 0} \frac{D(p\|q) - D(U(s+\eps,s) p\,\|\,U(s+\eps,s) q)}{\eps} \nonumber\\
    &=\left.-\frac d{d\eps}D(e^{\eps K_s}p\| e^{\eps K_s}q)\right|_{\eps=0}.
\end{align}
Therefore, we have an identity equivalent to the one in \eqref{eq:cont_KL_expansion}:
\begin{equation}
    \label{eq:time_dep_cont_KL_expansion}
    \MC(q) = D_{K_s}(p\| q) = -(K_sp)^\top\log \frac pq +p^\top \frac{K_sq}q.
\end{equation}
Alternatively, for fixed $q$, the instantaneous mismatch cost at time $t$ is 
\begin{equation}
    \MC(q) = -\frac d{dt} D(U(s+t, s)p\,\|\,q) + p^\top \frac{K_sq}q,
\end{equation}
since $\dot q=0$.

Consider a continuous process over an interval $[0, T]$ with a time-scaled rate matrix $K_t = w(t)K$, where $K$ is a time-independent rate matrix. 
In addition, \eqref{eq:time_dep_cont_KL_expansion} is linear in $K_t$, so
\begin{equation}
    D_{K_t}(p\|q) = w(t)D_{K}(p\|q).
\end{equation}
Define the trajectory weighted average to be
\begin{equation}
    \q_w:=\frac1W \int_0^T w(t)p_tdt,
\end{equation}
where $p_t = U(t, 0)p_0$ and $W:=\int_0^T w(t)dt$. Henceforth, all integral bounds are to be over $[0,T]$ unless stated otherwise.
Repeating the steps of the previous section, 
\begin{align}
    \Delta\MC_T &= \int w(t) (Kp_t)^\top\log \frac q{\q_w}  \nonumber \\
    & +\int w(t) p_t^\top \left(\frac{Kq}q - \frac{K\q_w}{\q_w}\right)dt,
\end{align}
we get the that the first term is
\begin{equation}
     \int w(t)(Kp_t)^\top\log \frac{q}{\q_w}dt  =-W (K\q_w)^\top \log \frac{\q_w}{q}.
\end{equation}
Similarly, the second term becomes
\begin{equation}
    \int w(t)p_t^\top \left(\frac{Kq}{q} - \frac{K\q_w}{\q_w}\right)dt = W \q_w^\top \frac{Kq}{q}.
\end{equation}
Thus, the difference is
\begin{equation}
    \Delta\MC_T =  
    W\left[ \q_w^\top \frac{Kq}{q} - (K\q_w)^\top \log \frac{\q_w}{q} \right],
\end{equation}
which gives the minimality of $\q_w$ via
\begin{equation}
    \MC_T(q)-\MC_T(\q_w) = WD_K(\q_w\| q)\geq 0.
\end{equation}

Plugging $\q_w$ back into the total mismatch cost, and since $p_t=U(t,0)p_0$, we have $\dot p_t = K_t p_t = w(t)Kp_t$,
\begin{align}
    \MC_T(\q_w) = &-\int_0^T \dot p_t^\top\log \frac{p_t}{\q_w} dt \nonumber \\
    &+ \int_0^Tw(t)p_t^\top \frac{K\q_w}{\q_w}\ dt.
\end{align}
Since the second term goes to zero, 
\begin{align}
    \MC_T(\q_w) = &-\int_0^T \dot p_t^\top\log p_t  dt
    + W (K\q_w)^\top \log \q_w.
\end{align}
Thus, the value of the MMC over time interval $[0, T]$ at the optimal weighted prior $\q_w$ is
\begin{equation}
    \MC_T(\q_w) =  S(p_T) -S(p_0) + W (K\q_w)^\top \log \q_w.
\end{equation}
Alternatively, we have
\begin{align}
    \MC_T(\bar q_w)
    &= -\int_0^T \dot p_t^\top\log \frac{p_t}{\q_w} dt \nonumber\\
    &=D(p_0\|\q_w)-D(p_T\|\q_w).
\end{align}

\section{Asymptotic behavior}\label{ap:asymp}
\noindent\textbf{Proposition.} Let $p_0\in\Delta_\X$ be an initial distribution and $G$ be an irreducible, but not necessarily aperiodic, finite stochastic matrix. Then, in the long $N$ limit, the trajectory-averaged optimal prior $\q\to\pi$, where $\pi$ is the unique stationary distribution of $G$.

\begin{proof}
Let $G$ be as in the proposition statement, so we have a unique stationary distribution $\pi$ such that $G\pi=\pi$. Notice that the following sum telescopes:
\begin{equation}
    (\mathbb{I}-G)\q = \frac1N\sum_{n=0}^{N-1} (p_n-Gp_n)=\frac1N(p_0-p_N).
\end{equation}
Thus, since $G$ is stochastic and $p_0\in\Delta_\X$,
\begin{equation}
    \|(\mathbb{I}-G)\q\|_1\leq\frac1N(\|p_0\|_1+\|p_N\|_1)=\frac2N.
\end{equation}
Therefore,
\begin{equation}
    \lim_{N\to\infty}\|(\mathbb{I}-G)\q\|_1=0.
\end{equation}
The sequence $(\q)$ lies in the compact probability simplex, so we may take a convergent subsequence $\q_{N_k}\to \q_\infty$, and, by continuity, we have
\begin{equation}
    (\mathbb{I}-G)\q_\infty=\lim_{k\to\infty}(\mathbb{I}-G)\q_{N_k}=0.
\end{equation}
Irreducibility implies that $\mathrm{ker}(\mathbb{I}-G)=\mathrm{span}\{\pi\}$, so since $\|\q_\infty\|_1=1$, we have $\q_\infty=\pi$. Thus, in the long-$N$ limit, we recover the stationary state as the optimal prior $\q_\infty$. The rate of convergence therefore goes as $O(N^{-1})$. 
\end{proof}

\noindent\textbf{Proposition.}
Assume now that $G$ is additionally aperiodic. The asymptotic behavior of the MMC is such that
\begin{equation}
    \lim_{N\to\infty}\MC_N(\q)=D(p_0\|\pi).
\end{equation}
In addition, the convergence rate is $O(N^{-1})$. 

\begin{proof}
As a shorthand, denote the mean distribution over $N$ iterations as $\q$ and let $\eps=\frac1N(p_N-p_0)$, so $G\q=\q+\eps$. 
A first-order Taylor expansion of $S(\q)$ about $G\q$ gives
\begin{equation}
    S(\q)-S(\q+\eps) = -\nabla S(\q+\eps)^\top\eps +O(\|\eps\|^2).
\end{equation}
Using $\nabla S(p) = - (1 + \log p)$,
\begin{equation}
    S(\q) - S(\q+\eps)=\eps^\top(1+\log (\q+\eps)) +O(\|\eps\|^2).
\end{equation}
However, since both $p_0$ and $p_N$ are normalized, $\eps^\top\mathbf{1}=0$ independent of $N$. Also, we may also use the first-order expansion $\log(\q+\eps)=\log \q+\eps/\q +O(\|\eps^2\|)$ to get 
\begin{equation}
    S(\q) - S(\q+\eps) = \eps^\top\log \q + O(\|\eps\|^2).
\end{equation}
Multiplying by $N$ and substituting into \eqref{eq:mmc}, we get
\begin{align}
    \MC_N^*
    &=S(p_N)-S(p_0) +(p_N-p_0)^\top\log \q + O(N\|\eps\|^2)\nonumber\\
    &=D(p_0\|\q)-D(p_N\|\q) + O(N^{-1}),
\end{align}
where we used $\|\eps\|=\frac1N\|p_N-p_0\| = O(N^{-1})$.
By the preceding proposition, $\q\to \pi$. Also, aperiodicity of $G$ gives that the sequence of distributions $p_N\to\pi$. Now, since irreducibility implies that the stationary distribution $\pi$ has full support, the KL divergence is continuous in a neighborhood of $\pi$. Consequently, we have both
\begin{equation}
    D(p_0\|\q)\to D(p_0\|\pi)
\end{equation}
and
\begin{equation}
    D(p_N\|\q)\to D(\pi\|\pi)=0.
\end{equation}
The $O(N^{-1})$ remainder also vanishes, and so
\begin{equation}
    \MC_\infty^* := \lim_{N\to\infty}\MC_N^* = D(p_0\|\pi),
\end{equation}
completing the first part of the proof. 

To establish convergence, we note that the Perron–Frobenius theorem implies that $G$ has a single eigenvalue of magnitude $1$. Now, note that 
\begin{equation}
    \pi\mathbf{ 1}^\top(p_0-\pi) = 0,
\end{equation}
so from the triangle inequality,
\begin{align}
    \|\q-\pi\|
    &\leq  \frac1N\sum_{n=0}^{N-1}\|G^np_0-\pi\|\nonumber\\
    &= \frac1N\sum_{n=0}^{N-1}\|(G^n - \pi\mathbf{1}^\top)(p_0-\pi)\|\nonumber\\
    &\leq  \frac{\|p_0-\pi\|}{N}\sum_{n=0}^{N-1}\|G^n - \pi\mathbf{1}^\top\|_{\mathrm{op}}.
\end{align}
Here, $\pi\mathbf{1}^\top$ picks out the subspace of eigenvalue 1 from $G$, so let $\lambda$ be the largest eigenvalue of $G$ such that $|\lambda|<1$. Then,
\begin{equation}
    \|\q-\pi\|\leq \frac CN \sum_{n=0}^{N-1}|\lambda|^n \leq \frac{C}{N(1-|\lambda|)}=O(N^{-1}),
\end{equation}
so convergence is upper-bounded by $O(N^{-1})$.
\end{proof}

\subsection{Upper bound on universal dissipation}\label{ap:upper}
\noindent\textbf{Proposition.} The universal dissipation $\MC_N^*$ is upper bounded in terms of the state space size $|\mathcal{X}|$ and number of iterations $N$:
\begin{equation}
\MC_N^* \leq 2\log |\mathcal{X}| + \log N + 1.
\end{equation}

\begin{proof}
Consider the Fannes--Audenaert inequality \cite{Audenaert_2007} for two probability distributions $p,q$ of dimension $|\mathcal{X}|$:
\begin{equation}
    |S(p)-S(q)| \leq T\log(|\mathcal{X}|-1)+H_b(T),
\end{equation}
where $T=\frac12|p-q|_1$, and $H_b(T)$ is the binary entropy of $(T,1-T)$. Recall
\begin{equation}
    \MC_N^* = S(p_N)-S(p_0)+N[S(\q)-S(G\q)].
\end{equation}
From \ref{ap:asymp},
\begin{equation}
    |\q-G\q|_1 =\frac1N|p_N-p_0|_1 \leq\frac2N.
\end{equation}
Thus $T\leq 1/N$. For $N\geq 2$, $H_b(T)$ is monotonically increasing over the relevant range, giving
\begin{equation}
    |S(\q)-S(G\q)| \leq \frac1N\log(|\mathcal{X}|-1)+H_b(1/N).
\end{equation}
Using $\log(|\mathcal{X}|-1)\leq\log |\mathcal{X}|$ and the bound
\begin{equation}
    H_b(x)\leq -x\log x+x,
\end{equation}
we obtain
\begin{equation}
    H_b(1/N)\leq\frac{\log N+1}{N}.
\end{equation}
Therefore,
\begin{equation}
    N|S(\q)-S(G\q)| \leq \log |\mathcal{X}|+\log N+1.
\end{equation}
Finally, since $S(p_N)\leq\log |\mathcal{X}|$ and $S(p_0)\geq0$, we have $S(p_N)-S(p_0)\leq\log |\mathcal{X}|$. Combining these inequalities gives
\begin{equation}
    \MC_N^*\leq2\log |\mathcal{X}|+\log N+1.
\end{equation}
\end{proof}

\subsection{Why the lower bound goes to $D(p_0\|\pi)$ in large $N$ limit}\label{ap:intuitive}
A main result of this letter is that if $G$ is irreducible and aperiodic, and therefore has a unique stationary distribution $\pi$, then in the large $N$ limit, the minimal periodic mismatch cost is $D(p_0\|\pi)$.

To understand this result intuitively for $\pi$ being equilibrium, consider a CTMC with time-homogeneous rate matrix $K$ modeling thermal relaxation. Then for any choice $\tau$ we have $p(t+\tau) = e^{K\tau}p(t)$. Therefore we have a periodic process with $G=e^{K\tau}$ and stroboscopic distributions $p_n:=p(n\tau)$. We know the actual prior for thermal relaxation is $\pi$, the equilibrium. The one-step MMC for thermal relaxation is $D(p_0\|\pi)$, which is all of the entropy production since residual cost is zero at equilibrium. Additionally, in \cite{yadav2026entropy} authors show that fine graining in time cannot decrease the MMC. Thus, the periodic MMC under the map $G$ will equal the one-step MMC. Now, the periodic MMC with $\pi$ as the prior is 
\begin{align}
    \MC_N(\pi) &= \sum_{n=0}^{N-1} D(p_n\|\pi) - D(Gp_n\|\pi)\\ 
    &= D(p_0\|\pi)-D(p_N\|\pi)\geq 0 
\end{align}
by the data processing inequality for KL divergence. Taking the limit as $N\to \infty$ we get $D(p_0\|\pi)$. Now consider any other prior $q\neq \pi$. Since $G^n p_0 \to \pi$ asymptotically, each summand approaches $D(\pi\|q) - D(\pi\|Gq)$ as  $n \to \infty$. By the data processing inequality for KL divergence applied to $G\pi = \pi$, this limit is strictly positive for any $q \neq \pi$. Since the terms of the sum do not vanish, the series diverges, giving, on $\Delta_\X\setminus\{\pi\}$
\begin{equation}
    \lim_{N\to\infty}\MC_N = +\infty,
\end{equation}
whereas using $\pi$ we obtain a finite answer. Therefore no prior other than $\pi$ can have lower MMC, and the minimal periodic mismatch cost in the $N\to\infty$ limit is $D(p_0\|\pi)$.

To understand this for an NESS, the argument above does not work directly because maps $G$ that model an NESS have nonzero probability currents at the fixed point, which differs from the $G$s mentioned above. Note, though, that we can always find a Hamiltonian $H$ for which the fixed point $\pi$ which was originally an NESS is equilibrium. This will certainly change the rate matrix $K$, and the  entropy flow $f(x)$ from \eqref{eq:cost}, but the lower bound does not depend on any of these. Therefore it will be a lower bound on the EP as if $\pi$ were equilibrium, even though it is an NESS.

\section{Spin system details}\label{ap:spins}
The energy change for flipping spin $i$ in configuration $\bs\sigma$ is
\begin{equation}
    \Delta_iH(\bs\sigma)=H(\bs\sigma^{(i)})-H(\bs\sigma)=\frac{2\bs\sigma_i}{n}\Bigl(J\sum_{j\neq i}\bs\sigma_j\Bigr),
\end{equation}
and only single-spin-flip transitions are allowed. With multiple reservoirs at inverse temperatures $\beta_\nu$, every flip has one channel per reservoir, so the observable rate matrix is $K=\sum_\nu K^{(\nu)}$ \cite{vandenbroeck2015ensemble}, and the distribution evolves as $p_t=e^{Kt}p_0$. Each channel obeys local detailed balance,
\begin{equation}
    \frac{K^{(\nu)}(\bs\sigma^{(i)}|\bs\sigma)}{K^{(\nu)}(\bs\sigma|\bs\sigma^{(i)})}=e^{-\beta_\nu\Delta_iH(\bs\sigma)},
\end{equation}
i.e.\ each reservoir pushes $p$ toward its own Boltzmann distribution, and we take the symmetric form $K^{(\nu)}(\bs\sigma^{(i)}|\bs\sigma)=\gamma_\nu e^{-\beta_\nu\Delta_iH(\bs\sigma)/2}$ with $\gamma_\nu>0$ \cite{mamedespins2026}. Differing temperatures violate global detailed balance, so the stationary state carries nonzero currents and is a NESS (Fig.~\ref{fig:vis}). The total EP over $[0,T]$ is computed from the continuous-time form of \eqref{eq:cost_form1}, with $F$ the expected entropy flow into the reservoirs (SM~\ref{ap:heat_flow}).

\begin{figure}
    \centering
    \includegraphics[width=1\linewidth]{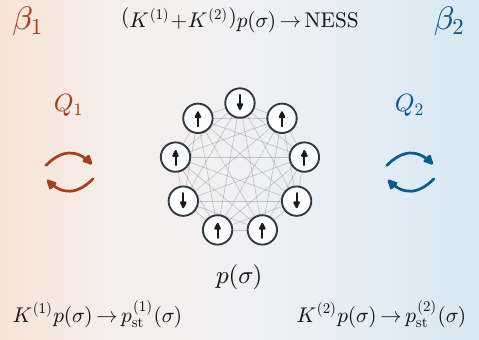}
    \caption{Spins coupled to baths at inverse temperatures $\beta_1,\beta_2$, exchanging heats $Q_1,Q_2$. Each reservoir pushes $p(\bs\sigma)$ toward its own Boltzmann distribution, so the combined stationary state carries nonzero current.}
    \label{fig:vis}
\end{figure}

\section{Speed limits and activity}\label{ap:speedlimit}
The speed-limit theorems \cite{speedlimits_ito, speedlimits_saito, speedlimits_fisher} bound the EP of a CTMC over $[0,T]$ in terms of information-theoretic properties of the trajectory and the rates. The bound plotted in Fig.~\ref{fig:placeholder} is
\begin{equation}
    \Sigma_T\geq 2W_1(p_0,p_T)\operatorname{arctanh}\!\bigl(W_1(p_0,p_T)/2T\langle\mu\rangle_T\bigr),
\end{equation}
where $W_1$ is the 1-Wasserstein distance (SM~\ref{ap:wass1}) and $2\langle\mu\rangle_T$ is the arithmetic-mean dynamical activity over $[0,T]$, roughly the total intensity of all jumps \cite{speedlimits_ito}; details on computing the activity are in SM~\ref{ap:activity}.

\subsection{Entropy flow for spin systems}\label{ap:heat_flow}
Let $\bs{\sigma}\to \bs{\sigma}^{(i)}$ be an Ising flip of spin $i$ mediated by bath $\nu$. Since the Hamiltonian is time-independent and there is no work performed during the flip, the first law gives 
\begin{equation}
    \Delta_i H(\bs{\sigma}) = H(\bs{\sigma}^{(i)}) - H(\bs{\sigma}) = q_i^{(\nu)}(\bs{\sigma}),
\end{equation}
where $q_i^{(\nu)}(\bs{\sigma})$ is the heat absorbed by the system from bath $\nu$ during this flip. Ensemble averaging, note that
\begin{equation}
    K^{(\nu)}(\bs{\sigma}^{(i)}|\bs{\sigma}) p_t(\bs{\sigma})
\end{equation}
is the probability per unit time of observing the transition $\bs{\sigma}\to \bs{\sigma}^{(i)}$ through reservoir $\nu$. Here, $p_t=e^{Kt}p_0$, where $K=\sum_\nu K^{(\nu)}$. Such a transition carries energy $\Delta_i H(\bs{\sigma})$, so define the state-dependent heat rate for all possible states to be 
\begin{equation}
    g^{(\nu)}(\bs{\sigma}):= \sum_i K^{(\nu)}(\bs{\sigma}^{(i)}|\bs{\sigma}) \Delta_i H(\bs{\sigma})
\end{equation}
since we only consider single-spin flips to be valid transitions. Then, the expected heat current into the system from reservoir $\nu$ is the sum over all possible states 
\begin{equation}
    \dot Q_\nu (t) = \sum_{\bs{\sigma}\in\mathcal{X}}g^{(\nu)}(\bs{\sigma})p_t(\bs{\sigma}) = \left(e^{K^\top t} g^{(\nu)} \right)^\top p_0.
\end{equation}
Accordingly, we have the total expected heat absorbed over $\mathcal{T}=[0, T]$ to be
\begin{align}
    Q_\nu(p_0) 
&= \int_\mathcal{T} \dot Q_\nu(t) \, dt = \int_\mathcal{T} \left(e^{K^\top t} g^{(\nu)}\right)^\top p_0 \, dt,  \nonumber\\
&=\left[\int_\mathcal{T} e^{K^\top t} g^{(\nu)} \, dt\right]^\top p_0,
\end{align}
which is a linear functional of $p_0$ that we may calculate entirely deterministically once we know $K^{(\nu)}$ and $H$. Then, for initial distribution $p_0$, the expected entropy flow to the environment $\Phi$ is the linear functional
\begin{equation}
    \Phi(p_0) = -\sum_\nu \beta_\nu Q_\nu(p_0).
\end{equation}
Using  in place of $F(p_0)$ in equation \eqref{eq:cost_form1} gives the total entropy production over time interval $[0,T]$.

\subsection{Activity Computation}
\label{ap:activity}
The activity can be considered to be the kinetic measure of bidirectional transitions. Denote the set of all transition edges as $\mathcal{E}$. For an edge $e=(i\to j)$ connecting states $i$ and $j$, define one-way fluxes as 
\begin{equation}
    J_e^+:=K_{ji}p_i, \quad J_e^-:=K_{ij}p_j.
\end{equation}
Choosing the arithmetic mean $\mu_{A,e}$ as convention for measuring the bidirectional fluxes along edge $e$, the activity is defined to be
\begin{equation}
    \mathcal{A}(t) = 2\mu_A(t) = 2\sum_{e\in\mathcal{E}}\mu_{A, e}(t),
\end{equation}
which is equivalent to 
\begin{equation}
    \mathcal{A}(t) = \sum_{i\neq j} K_{ij}(t)p_j(t) = -\sum_i K_{ii}p_i(t).
\end{equation}
Then, the activity over the time interval $\mathcal{T}$ is 
\begin{equation}
    \mathcal{A}_T = \int_\mathcal{T} 2\mu_A(t)dt.
\end{equation}

\subsection{Wasserstein distance calculation}
\label{ap:wass1}
For two distributions $p$ and $q$, their 1-Wasserstein distance is defined to be  \cite{speedlimits_ito}
\begin{equation}
    W_1(p, q) = \operatorname*{inf}_J\sum_{e\in\mathcal{E}}|J_e|,
\end{equation}
where $J_e:=J_{ji} - J_{ij}$ is the net current over edge $e=(i\to j)$, and where the infimum is over all net edge currents $J$ satisfying $\mathrm{div} J = p - q$, i.e., finding the currents $J_e$ to place on the edges of the graph such that the divergence gives the net current from $p$ to $q$. We find this by solving the dual problem of finding 
\begin{equation}
    \max_f\langle f, p-q\rangle
\end{equation}
subject to the Lipschitz-continuity constraint 
\begin{equation}
    |f_i - f_j|\leq d(i,j)
\end{equation}
on every edge $(i, j)$, where $d(i,j)$ is the length of said edge, which in the case of the spin system under consideration is always 1 (the Hamming distance).

\section{DFAs}\label{ap:DFAs}
\subsection{DFA background and example}
\emph{Deterministic Finite Automata (DFAs)---}
Formally, a DFA is a 5-tuple $(R,\mathcal{A},r^\varnothing,R^a,\rho)$, where $R$ is a finite set of computational states, $\mathcal{A}$ is a finite input alphabet, $r^\varnothing\in R$ is the initial state, $R^a\subseteq R$ is the set of accepting states, and $\rho:R\times \mathcal{A} \rightarrow R$ is the deterministic transition function \cite{Ouldridge_2023}. Given an input string $\lambda=(\lambda_1,\ldots,\lambda_L)$, the automaton begins in $r^\varnothing$ and successively updates its state according to
\begin{equation}
r_{t+1}=\rho(r_t,\lambda_t).
\end{equation}

The goal of a DFA is language recognition. A string a accepted if after the program halts, it is in the accept state. The set of all accepted strings is called a language. One can model a DFA as halting as soon as the accept state is reached, or one can model it as having a fixed number of $N$ iterations, and checking whether after $N$ iterations, the DFA is in the accept state. We will use the latter in this paper. 

As an example, consider the DFA shown in Fig.~\ref{fig:DFAplot}. This DFA processes input strings from the alphabet $\mathcal{A}=\{0,1\}$ and accepts strings that are integer representations of numbers divisible by $3$. The automaton is initialized in the state $r^\varnothing$ with certainty and processes a string of symbols drawn i.i.d. from $\mathcal{A}$. We assume each symbol is generated with $\Pr(0)=p$ and $\Pr(1)=1-p$. Under these assumptions, the evolution of the DFA is a discrete-time Markov chain (DTMC) over the computational states $R=\{0,1,2\}$, with transition probabilities determined by the (deterministic) transition function $\rho$ together with the input probabilities. 

The dynamics are completely characterized by the initial distribution $p_{0,i}=\delta_{0,i}$ and the stochastic transition matrix
\begin{equation}
G=
\begin{bmatrix}
p & 1-p & 0 \\
1-p & 0 & p \\
0 & p & 1-p  \\
\end{bmatrix},
\end{equation}
More generally, every deterministic transition labeled by the symbol $0$ contributes a probability $p$ to the corresponding entry of $G$, while transitions labeled by $1$ contribute a probability $1-p$. Thus, although the DFA transition function is deterministic, the stochasticity of the input stream induces a Markov chain on the computational states.

\subsection{Construction of redundant DFAs}\label{ap:dfa_construction}
Starting from $\mathrm{D3}$, take each state with a self-loop, here $\mathbf0$ and $\mathbf2$. For $\mathbf0$, which maps to itself on reading $0$, introduce a new state $\mathbf0_a$ reached from $\mathbf0$ on a $0$, and let $\mathbf0_a$ map back to $\mathbf0$ on a $0$; both $\mathbf0$ and $\mathbf0_a$ map to $\mathbf1$ on a $1$. Repeat for $\mathbf2$. The resulting DFA recognizes the same language, has `cycle states', and has many-to-one transitions, so it is not invertible. Repeating twice more yields $\mathrm{D5},\mathrm{D7},\mathrm{D9}$, where the number counts total states. The initial state remains $\mathbf0$ while the accept set becomes $\{\mathbf0,\mathbf0_a,\ldots\}$.

\subsection{Proof of asymptotic minimal dissipation scaling monotonically with redundant cycles}
\label{ap:dfa_proofs}
\begin{proof}
Let $\widetilde{\mathcal{A}}$ be a redundant (non-minimal) DFA with states $\widetilde{\mathcal{X}}$ and $\mathcal{A}$ a coarser, equivalent (accepts the same language) DFA with states $\mathcal{X}$. Suppose there is a surjective deterministic map between the state spaces\footnote{Indeed, this will exist from how we construct our redundant DFAs by adding redundant states.} $c:\widetilde{\mathcal{X}}\to \mathcal{X}$, such that, symbol-by-symbol (i.e., reading a 0 bit vs reading a 1 bit) we have the equivariance relation between the transition graphs for each symbol $b\in\Lambda$:
\begin{equation}
    C\widetilde{T}_b = T_b C,
\end{equation}
where the matrix $C$ is given by indicating which states $y\in\widetilde{\mathcal{X}}$ are redundancies of $x\in\mathcal{X}$:
\begin{equation}
    C_{xy} =\mathbf{1}{\{c(y)=x\}}.
\end{equation}
Then, by definition we have $G$ as a convex sum over the per-symbol operators $T_b$, i.e., $G = \sum_b p_b T_b$. Consequently, for the stochastic operators we also have the equivariance relation 
\begin{equation}
    C\widetilde{G} = GC.
\end{equation}
Now, if the induced action of $C$ over distributions of states follows $\tilde{p}_0\mapsto p_0 = C\tilde{p}_0$, then, at every iteration $n$, we have
\begin{equation}
    C \tilde{p}_n=C\widetilde{G}^n \tilde{p}_0 = G^nC\tilde{p}_0 = G^np_0  = p_n.
\end{equation}
Now, assuming $G$ and $\tilde G$ have unique fixed points $\pi$ and $\tilde{\pi}$, then the above equivariance relation gives
\begin{equation}
    GC\tilde{\pi} = C\widetilde{G}\tilde{\pi} = C\tilde{\pi},
\end{equation}
so $C\tilde{\pi}$ is a fixed point of $G$, and, in particular, $C\tilde{\pi} = \pi$. Now, using the data-processing inequality for KL divergence with respect to $C$, we have
\begin{equation}
    D(\tilde{p}_0\| \tilde{\pi}) \geq D(C\tilde{p}_0\| C\tilde{\pi}) = D(p_0\|\pi).
\end{equation}

Now, to show the conditions under which the inequality is strict, let $\Gamma_x = C^{-1}(x)\subseteq\widetilde{\mathcal{X}}$ be the fiber of redundant states mapping to state $x\in\mathcal{X}$ in the coarse DFA under $C$. Because $X = c(\tilde{X})$ is a deterministic function of $\tilde{X}$, the KL chain rule gives
\begin{align}
    D(\tilde{p}_0\| \tilde{\pi}) - D(p_0\|\pi) = \sum_{x\in\mathcal{X}} p_0(x) D(\tilde{p}_0(\cdot|x)\| \tilde{\pi}(\cdot|x)).
\end{align}
Note that here, for 
$y\in\widetilde{\mathcal{X}}$, we have 
$\tilde{p}_0(y|x) = \frac{\tilde{p}_0(y, x)}{p_0(x)}$. Since $X$ is a deterministic function of $\tilde{X}$, if $y\in \Gamma_x$, then knowing 
$\tilde{X}=y$ automatically implies $X=x$, so the joint is equal to the marginal in this case: 
$\tilde{p}_0(y, x) = \tilde{p}_0(y)$.
From the chain rule, we have equality if and only if
\begin{equation}
    \tilde{p}_0(\cdot | x) = \tilde{\pi}(\cdot | x)
\end{equation}
for every $x$ with $p_0(x)>0$. Equivalently, within each occupied fiber,
\begin{equation}
    \frac{\tilde{p}_0(y)}{\tilde{\pi}(y)}
\end{equation}
must be constant over that fiber. Suppose the DFA has initial state $p_0 = \delta_{x_0}$ and the redundant DFA begins in $\tilde{p}_0=\delta_{y_0}$, where $Cy_0=x_0$. Then, $\tilde{p}_0(\cdot| x_0) = \delta_{y_0}$ from the conditional probability formula with $p_0(x_0)=1$. Thus, 
\begin{align}
    D(\tilde{p}_0\| \tilde{\pi}) - D(p_0\|\pi) 
    &= D(\delta_{y_0}\| \tilde{\pi}(\cdot|x_0)) \nonumber\\
    &= -\log \tilde{\pi} (y_0|x_0) \nonumber\\
    &= -\log\frac{\tilde{\pi}(y_0)}{\pi(x_0)}
\end{align}
since the marginal of $\tilde{\pi}$ is the same as the joint. Since
\begin{equation}
    \pi(x_0)
    = (C\tilde{\pi})(x_0)
    = \sum_{y\in \Gamma_{x_0}} \tilde{\pi}(y),
\end{equation}
we always have $\tilde{\pi}(y_0)\leq \pi(x_0)$. If $\tilde{\pi}(y)>0$ for all $y$, then if there are more than one states in the preimage of $x_0$ under $C$, $|\Gamma_{x_0}|>1$, we must have $\tilde{\pi}(y_0) < \pi(x_0)$. Thus, 
\begin{equation}
    D(\tilde{p}_0\| \tilde{\pi}) > D(p_0\|\pi)
\end{equation}
if and only if $|\Gamma_{x_0}|>1$, i.e., an increase in the bound is governed solely by its redundancy in the initial state.
\end{proof}

\end{document}